\documentclass[journal,12pt,onecolumn]{IEEEtran}
\usepackage[top=1in, bottom=1in, left=1in, right=1in]{geometry}
\ifCLASSOPTIONonecolumn
\usepackage{setspace}
\fi

\usepackage{tikz}
\usetikzlibrary{automata,positioning}
\usepackage[T1]{fontenc}% optional T1 font encoding
\usepackage{amsmath,amsfonts,amssymb,amsthm,amsbsy}
\usepackage{float}
\usepackage{textcomp}
\usepackage{cite}
\usepackage{bm}
\usepackage{algorithm}
\usepackage{xcolor}
\usepackage{arydshln}
\usepackage{multicol}
\makeatletter
\newsavebox\myboxA
\newsavebox\myboxB
\newlength\mylenA

\usepackage{balance}
\newcommand*\xoverline[2][0.70]{%
    \sbox{\myboxA}{$\m@th#2$}%
    \setbox\myboxB\null% Phantom box
    \ht\myboxB=\ht\myboxA%
    \dp\myboxB=\dp\myboxA%
    \wd\myboxB=#1\wd\myboxA% Scale phantom
    \sbox\myboxB{$\m@th\overline{\copy\myboxB}$}%  Overlined phantom
    \setlength\mylenA{\the\wd\myboxA}%   calc width diff
    \addtolength\mylenA{-\the\wd\myboxB}%
    \ifdim\wd\myboxB<\wd\myboxA%
       \rlap{\hskip 0.5\mylenA\usebox\myboxB}{\usebox\myboxA}%
    \else
        \hskip -0.5\mylenA\rlap{\usebox\myboxA}{\hskip 0.5\mylenA\usebox\myboxB}%
    \fi}
\makeatother
\usepackage[algo2e]{algorithm2e}
\usepackage[noend]{algpseudocode}
\newcommand\blfootnote[1]{%
	\begingroup
	\renewcommand\thefootnote{}\footnote{#1}%
	\addtocounter{footnote}{-1}%
	\endgroup}
\usepackage{lipsum}
\usepackage{authblk}
\makeatletter
\renewcommand\AB@affilsepx{ \protect\Affilfont} % typesets affiliations next to each other
\makeatother
\makeatletter
\def\BState{\State\hskip-\ALG@thistlm}
\makeatother
\usepackage{mathtools}

\theoremstyle{remark}

\usepackage[section]{placeins}
\usepackage{amsmath}

\usepackage[cmintegrals]{newtxmath}

\newtheorem{theorem}{\textbf{Theorem}}
\renewenvironment{proof}{{\textbf{Proof.}}}{}
\begin{document}

	\title{Bit-Interleaved Coded Multiple Beamforming in Millimeter-Wave Massive MIMO Systems }

	\small
	\author{Sadjad Sedighi, \textit{Student Member, IEEE,}}
	\author{Ender Ayanoglu, \textit{Fellow, IEEE}}
	\affil{CPCC, Dept of EECS, UC Irvine, Irvine, CA, USA,}
	\normalsize

	% make the title area
	\maketitle
	\begin{abstract}

 An analysis of the diversity gain for bit-interleaved coded multiple beamforming (BICMB) method in millimeter-wave (mm-Wave) massive multiple-input multiple-output (MIMO) systems is carried out for both the single-user and multi-user scenario. We show that the diversity gain is independent of the number of data streams and full spatial multiplexing order can be achieved in both scenarios. Also, we show that the diversity gain in the multi-user scenario is independent of the number of users in the system and only depends on the number of the remote antenna units (RAUs) at the transmitter side, when each user has only one RAU. The assumption here is that the channel state information (CSI) is known at both sides of the transmitter and the receiver and the number of antennas in each RAU goes to infinity. This latter assumption can be relaxed by a large number of antennas in each RAU, similar to the case for all massive MIMO research. Based on these assumptions, the diversity gain for the single-user scenario is $\frac{\left(\sum_{i,j}\beta_{ij}\right)^2}{\sum_{i,j}\beta_{ij}^2L_{ij}^{-1}}$  where $L_{ij}$ is the number of propagation paths and $\beta_{ij}$ is the large scale fading coefficient between the $i$th RAU in the transmitter and the $j$th RAU in the receiver. The diversity gain in the multi-user scenario for the $k$-th user is $\frac{M^2}{\sum_{j}L_{kj}^{-1}}$ where $M$ represents the number of RAUs at the transmitter. Simulation results show that when the perfect channel state information assumption is satisfied, the use of BICMB results in the diversity gain values predicted by the analysis.
	\end{abstract}
\vfill
	\blfootnote{This work was partially supported by NSF under Grant No. 1547155. This work was partially presented during the IEEE International Conference on Communications, May 21-23, 2019.}
	\pagebreak
	\section{Introduction}
Millimeter-wave (mm-Wave) massive multiple-input multiple-output (MIMO) will likely become an important part of the Fifth Generation (5G) communication systems. It enables us to increase the data rates and help in accommodating the billions of wireless devices whose numbers increase exponentially each year \cite{Rappaport2013,Swindle2014,Roh2014}. Despite of its substantial gains, mm-Wave massive MIMO brings challenges. Severe penetration loss and path loss in the mm-Wave signals comparing to signals in former and current cellular systems (e.g., 3G or LTE) are two of the challenges \cite{Haidar2014}.

{\color{black}
 One of the advantages of the mm-Wave frequencies is that they enable one to pack more antennas in the same area compared to a lower range of frequencies. This leads to highly directional beamforming and large-scale spatial multiplexing\footnote{In this paper, the terms "spatial multiplexing" or "spatial multiplexing
order" are used as in \cite{Paulraj2003} to describe the number of spatial subchannels. Note that this term is different from "spatial multiplexing gain" defined in \cite{Zheng2003}} in mm-Wave frequencies. The principles of beamforming are independent of the carrier frequency, but it is not practical to use fully digital beamforming schemes for massive MIMO systems \cite{Telatar1999,Shi2011,Dahrouj2010,Peel2005}. Power consumption and cost perspectives are the main obstacles due to the high number of radio frequency (RF) chains required for the fully digital beamforming, i.e., one RF chain per antenna element \cite{Doan2004}.} To address this problem, hybrid analog-digital processing of the precoder and combiner for mm-Wave communications systems is being considered \cite{Sohrabi2016,Ayach2012,Ayach2013,Weiheng2016,Singh2014,Zhang2015,XWu2018}, where \cite{XWu2018} proposes an algorithm to calculate the beamforming matrices in a closed form.

	Bit-interleaved coded modulation (BICM) was introduced as a way to increase the code diversity \cite{Caire1998,Zehavi1992}.  As stated in \cite{Akay2007}, bit-interleaved coded multiple beamforming (BICMB) has a substantial impact on the diversity gain performance of a MIMO system. Recently, the diversity gain in mm-Wave for both co-located and distributed systems is studied in \cite{Dian2018}. The authors of \cite{Dian2018} have shown that increasing the number of remote antenna units (RAUs) in the distributed system increases the diversity gain. In \cite{Sedighi2019}, using BICMB to increase the diversity gain is investigated for the single-user scenario. In this work, we extend the method we used in \cite{Sedighi2019} to a multi-user scenario where multiple users are being served with a BS. In Sections II and III, BICMB for the system under consideration is analyzed. We show that by using BICM in the system, one can achieve full spatial multiplexing without any loss in the diversity gain. That is, in Section III, we show that BICMB achieves full diversity order of $M_rM_tL$ and full spatial multiplexing order of $M_rM_tL$ in a special case when the number of propagation paths is constant for all paths between RAUs, i.e., $L_{ij}=L$ over the mm-Wave channel. We provide design criteria for the interleaver that guarantee full diversity and full spatial multiplexing. In Section IV, the system model and channel model are introduced. In this section, a hybrid beamforming method \cite{XWu2018} is used to eliminate the inter-user interference and to maximize the achievable rate. The difference of the method in \cite{XWu2018} with the other methods is its closed form where the precoder and combiner matrices can be calculated explicitly without any need for iteration. In Section V we use pairwise error probability (PEP) for convolutional coding in BICMB to find an upper bound for error probability.  Then we show that BICMB achieves full diversity order of $ML$ in a special case when the number of propagation paths is constant for all paths between RAUs, i.e., $L_{ij}=L$ over the mm-Wave channel. We provide design criteria for the interleaver that guarantee full diversity and full spatial multiplexing. Simulation results are provided in Section VI. Finally conclusion s are presented in Section VII.

\iffalse
By using BICMB, we showed in \cite{Sedighi2019} that for a single-user in mm-Wave massive MIMO systems, the diversity gain can be increased. However, the mm-Wave massive MIMO systems are built for multiple users. Thus, it can be expected to use the same idea to increase the diversity gain for a multi-user scenario.
\fi
We would like to emphasize that the asymptotical diversity analysis obtained in this paper is under the idealistic assumption of having perfect channel state information both at the transmitter and the receiver as done in similar works.

	\textit{Notation:}
Boldface upper and lower case letters denote matrices and column vectors, respectively.	The minimum Hamming distance
	between any two codewords in a convolutional code is defined as the free distance $d_{\text{free}}$. The symbol $N_s$ denotes
	the total number of symbols transmitted at a time.
	The minimum Euclidean distance between the two constellation points is given by $d_{\text{min}}$. The symbols $(.)^H, (.)^T, (.)^*,(\bar{.})$ and $\forall$ denote the Hermitian, transpose, complex conjugate, binary complement, and for all, respectively. $\mathcal{CN} (0, 1)$ denotes a circularly symmetric complex Gaussian random distribution with zero mean and unit variance. The expectation operator is denoted by $\mathbb{E}\left[.\right]$.  $[\mathbf{A}]_{ij}$ or $\mathbf{A}(i,j)$ gives the $(i, j)$-th entry of matrix $\mathbf{A}$. $\mathbf{A}(i,:)$ and $\mathbf{A}(:,j)$ represent the $i$-th row of the matrix $\mathbf{A}$ and $j$-th column of the matrix $\mathbf{A}$, respectively. Finally, diag$\{a_1 , a_2 , \dots , a_N \}$ stands for a diagonal matrix with diagonal elements $\{a_1 , a_2 ,\dots , a_N \}$.

	\section{System Model For Single-User Scenario}
	Consider a downlink single-user mm-Wave massive MIMO system as shown in Fig. \ref{SU_FC}. In this system, the transmitter sends $N_s$ data streams to a receiver. The transmitter is equipped with  $N_t^{RF}$ RF chains and $M_t$ RAUs, where each RAU has $N_t$ antennas, while at the receiver, the number of RF chains and RAUs is given by $N_r^{RF}$ and $M_r$, respectively. Each RAU at the receiver has $N_r$ antennas. When $M_t = M_r = 1$, the system reduces to a conventional co-located MIMO (C-MIMO) system.

	\begin{figure}[!t]
		\centering
        \includegraphics[width=0.8\textwidth]{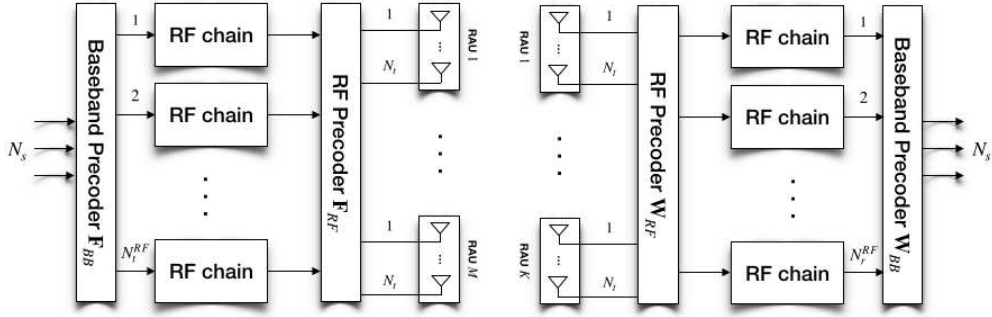}
		\caption{Block diagram of a mm-Wave massive MIMO system with distributed antenna arrays.}
		\label{SU_FC}
	\end{figure}
	
	The input to the system is $N_s$ data streams. The vector of data symbols to be transmitted by the transmitter at each time instant, $\mathbf{x} \in \mathbb{C}^{N_s\times1}$, can be expressed as
	\begin{align}
	\mathbf{x}=\left[x_1,...,x_{N_s}\right]^T,
		\end{align}
		where $E\left[\mathbf{x}\mathbf{x}^H\right]=\mathbf{I}_{N_s}$. The preprocessing at the baseband is applied by means of the matrix $\mathbf{F}_{\text{BB}} \in \mathbb{C}^{N_t^{RF} \times N_s}$. The last stage of data preprocessing is performed at RF, when beamforming is applied by means of phase shifters and combiners. A set of $M_tN_t$ phase shifters is applied to the output of each RF chain. As a result of this process, different beams are formed in order to transmit the RF signals. We can model this process with an $M_t N_t \times N_t^{RF} $ complex matrix $\mathbf{F}_{\text{RF}}$. Note that the baseband precoder $\mathbf{F}_{\text{BB}}$ modifies both amplitude and phases, while only phase changes can be realized by $\mathbf{F}_{\text{RF}}$ since it is implemented by using analog phase shifters.
		
We assume a narrowband flat fading channel model and obtain the received signal	as
	\begin{align}
	\mathbf{z}=\mathbf{H}\mathbf{F}_{\text{RF}}\mathbf{F}_{\text{BB}}\mathbf{x}+\mathbf{n},
	\end{align}
	where $\mathbf{H}$ is an $M_rN_r\times M_tN_t$ channel matrix with complex-valued entries and $\mathbf{n}$ is an $M_rN_r\times1$ vector consisting of i.i.d. $\mathcal{CN}(0,N_0)$ noise samples, where $N_0=\frac{N_t}{SNR}$. The processed signal is given by
	\begin{align}\label{rec_sig}
	\mathbf{y}=\mathbf{W}_{\text{BB}}^H\mathbf{W}_{\text{RF}}^H\mathbf{H}\mathbf{F}_{\text{RF}}\mathbf{F}_{\text{BB}}\mathbf{x}+\mathbf{W}_{\text{BB}}^H\mathbf{W}_{\text{RF}}^H\mathbf{n},
	\end{align}
	where $\mathbf{W}_{\text{RF}}$ is the $M_rN_r\times N_r^{RF}$ RF combining matrix, and $\mathbf{W}_{\text{BB}}$ is the $N_r^{(\text{RF})}\times N_s$ baseband combining matrix.
	
	The channel matrix $\mathbf{H}$ can also be written as
	\begin{align}
	\mathbf{H}=
	\begin{bmatrix}
	\sqrt{\beta_{11}}\mathbf{H}_{11}&\sqrt{\beta_{12}}\mathbf{H}_{12} & \dots& \sqrt{\beta_{1M_t}}\mathbf{H}_{1M_t} \\
	\sqrt{\beta_{21}}\mathbf{H}_{21}&\sqrt{\beta_{22}}\mathbf{H}_{22} & \dots& \sqrt{\beta_{2M_t}}\mathbf{H}_{2M_t} \\
	\vdots & \vdots & \ddots & \vdots  \\
	\sqrt{\beta_{M_r1}}\mathbf{H}_{M_r1} & \sqrt{\beta_{M_r2}}\mathbf{H}_{M_r2} & \dots & \sqrt{\beta_{M_rM_t}}\mathbf{H}_{M_rM_t}
	\end{bmatrix},
	\end{align}
	where $\beta_{ij}$, a real-valued nonnegative number, represents the large-scale fading effect between the $i$th RAU at the receiver and $j$th RAU at the transmitter. The normalized subchannel matrix $\mathbf{H}_{ij}$ is the MIMO channel between the $i$th RAU at the receiver and the $j$th RAU at the transmitter.
	
	Analytical channel models such as Rayleigh fading are not suitable for mm-Wave channel modeling. The reason for this is the fact that the scattering levels represented by these models are too rich for mm-Wave channels\cite{Ayach2012}. In this paper, the model is based on the Saleh-Valenzuela model that is often used in mm-Wave channel modeling \cite{HXu2002} and standardization \cite{Standard}. For simplicity, each scattering cluster is assumed to contribute a single propagation path. The subchannel matrix $\mathbf{H}_{ij}$ is given by
	\begin{align}\label{Hij}
	\mathbf{H}_{ij}=\sqrt{\frac{N_tN_r}{L_{ij}}} \sum_{l=1}^{L_{ij}}\alpha_{ij}^l\mathbf{a}_r(\theta_{ij}^{l}) \mathbf{a}_t^H(\phi_{ij}^{l}),
	\end{align}
	where $L_{ij}$ is the number of propagation paths and $\alpha_{ij}^{l}$ is the complex-gain of the $l$th ray which follows $\mathcal{CN}(0,1)$, the vectors $\mathbf{a}_r(\theta_{ij}^{l})$ and $ \mathbf{a}_t(\phi_{ij}^{l})$ are the normalized receive/transmit array response and $\theta_{ij}^{l} $ and $\phi_{ij}^{l}$ are its random azimuth angles of arrival and departure respectively.
	
	The uniform linear array (ULA) is employed by the transmitter and receiver in our study. For an $N$-element ULA, the array response vector is given by
	\begin{align}
	\mathbf{a}_{ULA}(\phi)=\frac{1}{\sqrt{N}}\left[1, e^{j\frac{2\pi}{\lambda}dsin(\phi)},\dots, e^{j(N-1)\frac{2\pi}{\lambda}dsin(\phi)}\right]^T,
	\end{align}
	where $\lambda$ is the wavelength of the carrier, and $d$ is the distance between neighboring antenna elements.

	We leverage both BICM and multiple beamforming to form BICMB \cite{Akay2007}.  An interleaver is used to interleave the output bits of a binary convolutional encoder. Then the output of the interleaver is mapped over a signal set $\chi \subset \mathbb{C}$ of size $|\chi|=2^m$ with a binary labeling map $\mu:\{0,1\}^m\rightarrow \chi$. The interleaver design has two criteria\cite{Akay2007}:
	\begin{enumerate}
		\item Consecutive coded bits are mapped to different symbols.
			\item Each subchannel should be utilized at least once within $d_{\text{free}}$ distinct bits among different codewords by using proper code and interleaver.
	\end{enumerate}
	Note that the free distance $d_{\text{free}}$ of the convolutional encoder should satisfy $d_{\text{free}}\geq N_s$ \cite{Akay2007}.

For mapping the bits onto symbols, Gray encoding is used. Also, we are using a Viterbi decoder at the receiver. The interleaver $\pi$ is used to interleave the code sequence $\underline{c}$. Then the output of the interleaver is mapped onto the signal sequence $\underline{x} \in \chi$.

The only beamforming constraint here is a total power constraint, because one can control both the amplitude and the phase of a signal. The total power constraint leads to a simple solution based on Singular Value Decomposition (SVD) \cite{Ayach2012}
	\begin{align}
	\mathbf{H}=\mathbf{U\Lambda V}^H=\left[\mathbf{u}_1 \mathbf{u}_2 \dots \mathbf{u}_{M_rN_r}\right]^H \mathbf{H} \left[\mathbf{v}_1 \mathbf{v}_2 \dots \mathbf{v}_{M_tN_t}\right],
	\end{align}
	where $\mathbf{U}$ and $\mathbf{V}$ are $M_rN_r\times M_rN_r$ and $M_tN_t\times M_tN_t$ unitary matrices, respectively, and $\mathbf{\Lambda}$ is an $ M_rN_r \times M_tN_t$ diagonal matrix with singular values of $\mathbf{H}$, $\lambda_i \in \mathbb{R}$, on the main diagonal with decreasing order.
By exploiting the optimal precoder and combiner, the system input-output relation in (\ref{rec_sig}) at the $k$th time instant can be written as
	\begin{align}
	\mathbf{y}_k=\left[\mathbf{u}_1 \mathbf{u}_2 \dots \mathbf{u}_{N_s}\right]^H \mathbf{H} \left[\mathbf{v}_1 \mathbf{v}_2 \dots \mathbf{v}_{N_s}\right] \mathbf{x}+ \left[\mathbf{u}_1 \mathbf{u}_2 \dots \mathbf{u}_{N_s}\right]^H\mathbf{n}_k,
	\end{align}
	\begin{align}\label{def_rec}
	y_{k,s}=\lambda_sx_{k,s} +n_{k,s}, \text{   \quad for } s=1,2,\dots,N_s.
	\end{align}
	
	\section{Diversity Gain and PEP Analysis For Single-User Scenario}
	In this section, we show that by using the BICMB analysis for calculating BER, the diversity gain becomes independent of the number of data streams.
	{\color{black}
		\begin{theorem}
			Suppose that $N_r\rightarrow \infty$ and $N_t \rightarrow \infty$. Then the bit interleaved coded distributed massive MIMO system can achieve a diversity gain of
		\begin{align}\label{su_dg}
		D_G=\frac{\left(\sum_{i,j}\beta_{ij}\right)^2}{\sum_{i,j}\beta_{ij}^2L_{ij}^{-1}}
		\end{align}
			$\text{ for } i=1,\dots,M_r \text{ and } j=1,\dots,M_t$.
	\end{theorem}}

	\begin{proof}
			We model the BICMB bit interleaver as $\pi: k' \rightarrow (k,s,i)$, where $k'$ represents the original ordering of the coded bits $c_{k'}$, $k$ represents the time ordering of the signals $x_{k,s}$ and $i$ denotes the position of the bit $c_{k'}$ on symbol $x_{k,s}$.
			
		We define $\chi_b^i$ as the subset of all signals $x\in\chi$. Note that the label has the value $b\in\{0,1\}$ in position $i$.
		
		Then, the ML bit metrics are given by using (\ref{def_rec}), \cite{Akay2007,Caire1998,Zehavi1992}
		\begin{align}\label{ML_bit}
		\gamma^i(y_{k,s},c_{k'})=\min_{x \in \chi^i_{c_{k'}}} \left|y_{k,s}-\lambda_sx\right|^2.
		\end{align}
		
		The receiver uses an ML decoder to make decisions based on
		\begin{align}\label{ML_dec}
		\mathbf{\underline{\hat{{c}}}}=\text{arg}\min_{\mathbf{\underline{c}}\in\mathcal{C}}\sum_{k'}\gamma^i(y_{k,s},c_{k'}).
		\end{align}
		
		Assume that the code sequence $\underline{\text{c}}$ is transmitted and $\underline{\hat{\text{c}}}$ is detected. Then by using (\ref{ML_bit}) and (\ref{ML_dec}), the pairwise error probability (PEP) of $\underline{\text{c}}$ and $\underline{\hat{\text{c}}}$ given channel state information (CSI) can be written as \cite{Akay2007}
		{\color{black}
			\begin{align}
			P(\underline{\text{c}}\rightarrow \underline{\hat{\text{c}}}|\mathbf{H})
			=P\left(\sum_{k'} \min_{x \in \chi^i_{c_{k'}}} |y_{k,s}-\lambda_sx|^2 \geq \sum_{k'} \min_{x \in \chi^i_{\hat{c}_{k'}}} |y_{k,s}-\lambda_sx|^2\right),
			\end{align}}
		where $s \in \{1,2,\dots,N_s\}$.
	
		Note that in a convolutional code, the Hamming distance between $\underline{\text{c}}$ and $\hat{\underline{\text{c}}}$, $d(\underline{\text{c}}-\underline{\hat{\text{c}}})$ is at least $d_{\text{free}}$. In this work we assume for PEP analysis $d(\underline{\text{c}}-\underline{\hat{\text{c}}})=d_{\text{free}}$.
		
		For the $d_{\text{free}}$ bits, let us denote
		\begin{align}
		\tilde{x}_{k,s}=\text{arg}\min_{x \in \chi^i_{{c}_{k'}}}\left|y_{k,s}-\lambda_sx\right|^2\\
		\hat{x}_{k,s}=\text{arg}\min_{x \in \chi^i_{\bar{c}_{k'}}}\left|y_{k,s}-\lambda_sx\right|^2
		\end{align}
		
		By using the trellis structure of the convolutional codes \cite{Akay2007}, one can write
		\begin{align}
		P(\underline{\text{c}}\rightarrow\underline{\hat{\text{c}}}|\mathbf{H})\leq Q\left(\sqrt{\frac{d_{\text{min}}^2\sum_{s=1}^{N_s}\alpha_s\lambda_s^2}{2N_0}}\right)
		\end{align}
		where $\alpha_s$ is a parameter that indicates how many times subchannel $s$ is used within the $d_{\text{free}}$ bits under consideration, and  $\sum_{s=1}^{N_s}\alpha_s=d_{\text{free}}$.
		The bound $Q(x)\leq\frac{1}{2}e^{-\frac{x}{2}}$ can be used to upper bound the PEP as
		\begin{align}\label{PEP_exp}
		P(\underline{\text{c}}\rightarrow\underline{\hat{\text{c}}})= \mathit{E}\left[P(\underline{\text{c}}\rightarrow\underline{\hat{\text{c}}}|\mathbf{H})\right]
		\leq \mathit{E}\left[\frac{1}{2}\text{exp}\left(\frac{-d_{\text{min}}^2\sum_{s=1}^{N_s}\alpha_s\lambda_s^2}{4N_0}\right)\right].
		\end{align}
		
		Let us denote $\alpha_{\text{min}} =\text{min }\{\alpha_s :s=1,2,\dots,N_s\}$. Then
		\begin{align}\label{sing_values}
		\frac{\left(\sum_{s=1}^{N_s}\alpha_s\lambda_s^2\right)}{N_s}\geq \frac{\left(\alpha_{\text{min}}\sum_{s=1}^{N_s}\lambda_s^2\right)}{N_s} \geq \frac{\left(\alpha_{\text{min}}\sum_{s=1}^{L_t}\lambda_s^2\right)}{L_t}.
		\end{align}
		There are only $L_t$ non-zero singular values \cite{Dian2018}.
				
		Let us define
		\begin{align} \label{theta}
		\Theta \triangleq \sum_{s=1}^{N_s}\lambda_s^2=||\mathbf{H}||_F^2=\sum_{i=1}^{M_r}\sum_{j=1}^{M_t}\beta_{ij}||\mathbf{H}_{ij}||_F^2.
		\end{align}
		
		Theorem 3 in \cite{Ayach2012} implies that the singular values of $\mathbf{H}_{ij}$ converge to $\sqrt{\frac{N_rN_t} {L_{ij}}}\left|\alpha_{l}^{ij}\right|$ in descending order. By using the singular values of $\mathbf{H}_{ij}$, (\ref{theta}) can be rewritten as

\begin{align}\label{Psi_def}
		\Theta=\sum_{i=1}^{M_r}\sum_{j=1}^{M_t}\beta_{ij}||\mathbf{H}_{ij}||_F^2= N_rN_t\sum_{i=1}^{M_r}\sum_{j=1}^{M_t}\underbrace{\frac{\beta_{ij}}{L_{ij}}\sum_{l=1}^{L_{ij}}\left|\alpha_{ij}^l\right|^2}_{\Psi_{ij}}.
\end{align}
		
		Note that the random vairable $\sum_{l=1}^{L_{ij}}\left|\alpha_{ij}^l \right|$ has a $\chi$-squared distribution with $2L_{ij}$ degrees of freedom, or equivalently a Gamma distribution with shape $L_{ij}$ and scale 2, denoted $\mathcal{G}(L_{ij},2)$. Then, since $\beta_{ij} L_{ij}^{-1}>0$, $\Psi_{ij} \sim \mathcal{G}(L_{ij},2\beta_{ij}L_{ij}^{-1})$ \cite{Hogg1978}. One can use the Welch-Satterthwaite equation to calculate an approximation to the degrees of freedom of $\Theta$ (i.e., shape and scale of the Gamma distribution) which is a linear combination of the independent random variables $\Psi_{ij}$ \cite[p.4.1-1]{Satterth1946},\cite{Massey}
\begin{align}\label{shape}
		k&=\frac{\left(\sum_{i,j}\theta_{ij}k_{ij}\right)^2}{\sum_{i,j}\theta_{ij}^2k_{ij}}=\frac{\left(\sum_{i,j}\beta_{ij}\right)^2}{\sum_{i,j}\beta_{ij}^2L_{ij}^{-1}},\\\label{scale}
		\theta&=\frac{\sum_{i,j}\theta_{ij}^2k_{ij}}{\sum_{i,j}\theta_{ij}k_{ij}}=\frac{\sum_{i,j}\beta_{ij}^2L_{ij}^{-1}}{\sum_{i,j}\beta_{ij}}.
\end{align}

		Using (\ref{PEP_exp}), (\ref{sing_values}), and (\ref{theta}), the PEP is upper bounded by
		\begin{align}\label{PEP_exp2_single_user}
		P(\underline{\text{c}}\rightarrow\underline{\hat{\text{c}}})\leq\frac{1}{2} \mathit{E}\left[\text{exp}\left(\frac{-d_{\text{min}}^2\alpha_{\text{min}}N_s}{4N_0L_t}\Theta\right)\right],
		\end{align}
		which is the definition of the moment generating function (MGF)\cite{Bulmer1965} for the random variable $\Theta$. By using the definition, (\ref{PEP_exp2_single_user}) can be written as
		\begin{align}\nonumber
		P(\underline{\text{c}}\rightarrow\underline{\hat{\text{c}}})=&g(d, \alpha_{\text{min}}, \chi)\\
		\leq&\frac{1}{2} \left(1+ \theta \frac{d_{\text{min}}^2\alpha_{\text{min}}N_sN_t}{4L_t}SNR\right)^{-k}\label{PEP_MGF1}\\
		\approx&\frac{1}{2}\left(\theta\frac{d_{\text{min}}^2\alpha_{\text{min}}N_sN_t}{4L_t}SNR\right)^{-k} \label{PEP_MGF2}
		\end{align}
		for high SNR.  The function $g(d, \alpha_{\text{min}}, \chi)$ denotes the PEP of two codewords with $d(\underline{\text{c}}-\underline{\hat{\text{c}}}) = d$, with $\alpha_{\text{min}}$ corresponding to $\underline{\text{c}}$ and $\underline{\hat{\text{c}}}$, and with constellation $\chi$.
		In (\ref{PEP_MGF1})  $\theta$ and $k$ are defined as (\ref{shape}) and (\ref{scale}).
		
		In BICMB, $P_b$ can be calculated as \cite{Akay2007}
		\begin{align}\label{BICMB_Pb_single_user}
		P_b \leq \frac{1}{k_c}\sum_{d=d_{\text{free}}}^{\infty}\sum_{i=1}^{W_I(d)}g(d,\alpha_{min}(d,i),\chi),
		\end{align}
		where $W_I(d)$ denotes the total input weight of error events at Hamming distance $d$.	Following (\ref{PEP_MGF2}) and (\ref{BICMB_Pb_single_user})
		\begin{align} \label{SU_bound}
		P_b \leq\frac{1}{k_c} \sum_{d=d_{\text{min}}}^{\infty}\sum_{i=1}^{W_I(d)} \frac{1}{2}\left(\theta\frac{d_{\text{free}}^2\alpha_{\text{min}}N_sN_t}{4L_t}SNR\right)^{-k}.
		\end{align}
		
		The SNR component has a power of $-k$ for all summations. Hence, BICMB achieves full diversity order of
		\begin{align}\label{su_dg2}
		D_G=k=\frac{\left(\sum_{i,j}\beta_{ij}\right)^2}{\sum_{i,j}\beta_{ij}^2L_{ij}^{-1}}
		\end{align}
	which is independent of the number of spatial streams transmitted.
	\end{proof}

	\newtheorem{remarks}{\textbf{Remark}}
	\begin{remarks}
		Under the case where $N_t$ and $N_r$ are large enough and assuming that $L_{ij}=L$ and $\beta_{ij}=\beta$ for any $i$ and $j$, it can be seen easily that the distributed massive MIMO system can achieve a diversity gain
		\begin{align}
	    D_G=L_t=M_rM_tL.
		\end{align}
		
	\end{remarks}

	\begin{remarks}
		Theorem 1 implies that the diversity gain is independent of the number of data streams, i.e., the transmitter can send the maximum number of data streams $N_s\leq L_t$, and still get the same diversity gain. This will be illustrated in Section IV.
	\end{remarks}

	\section{System Model for Multi-user Scenario}
		Consider a downlink multi-user massive MIMO system as shown in Fig. \ref{mu_fc}. The antenna array at the base station (BS) consists of $N_{t}^{\text{RF}}$ RF chains and $M$ RAUs, each of which has $N_t$ antennas. There are $K$ different mobile stations (MS) and each one of them is equipped with $N_r$ antennas and $N_r^{\text{RF}}$ RF chains. The BS transmits $KN_s$ data streams and each MS receives its $N_s$ data streams. We constrain the number of RF chains in order to reduce the hardware complexity of the massive MIMO system. This constraint for the BS is $KN_s\leq N_{t}^{RF}\ll N_t$ and $N_s\leq N_{r}^{RF}\ll N_r$ for each MS.
	
		\begin{figure}[!t]
		\centering
        \includegraphics[width=0.8\textwidth]{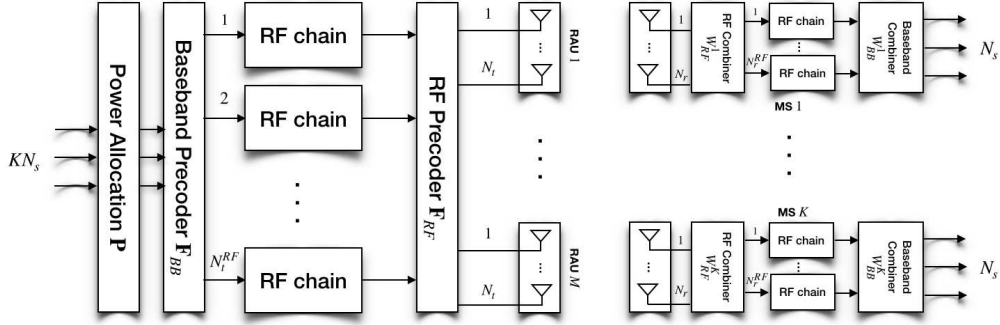}
		\caption{Block diagram of a multi-user mm-Wave massive MIMO system with distributed antenna arrays.}
		\label{mu_fc}
	\end{figure}

	We denote the RF precoder $\mathbf{F}_{RF}$ by an $MN_t\times N_t^{\text{RF}}$ matrix and the baseband precoder  $\mathbf{F}_{BB}$ by an $N_t^{\text{RF}}\times KN_s$ matrix. At the BS, the transmitted symbols of $K$ users first go through a power allocation matrix $\mathbf{P}$ which is  $KN_s\times KN_s$  and $||\mathbf{P}||_F^2=P_t$. Since the RF precoder matrix only changes the phase of the input signal, its magnitude is constant, i.e., $\left|\mathbf{F}_{RF}(i,j)\right|=\frac{1}{\sqrt{N_t}}$. Also, because of the power constraint at the BS, we need to satisfy $|| \mathbf{F}_{RF} \mathbf{F}_{BB}||_F^2=KN_s$. We assume that the CSI is known at both transmitter and receiver. We employ a narrowband flat fading channel model for CSI. The received signal at the $k$-th MS after combining is given by
	\begin{align}\label{rec_sig0}
	    \mathbf{y}_k={\mathbf{W}^k}^H_{BB}{\mathbf{W}^k}^H_{RF}{\mathbf{H}}_{k}{\mathbf{F}}_{RF}{\mathbf{F}}_{BB}\mathbf{P}\mathbf{x}+{\mathbf{W}^k}^H_{BB}{\mathbf{W}^k}^H_{RF}\mathbf{n}_k
	\end{align}
	where $k \in \{1,2,\dots,K\}$. The channel matrix corresponding to the $k$-th MS $\mathbf{H}_k$ is $N_r\times MN_t$  and $\mathbf{n}_k$ is a $N_r\times1$ vector consisting of i.i.d. $\mathcal{CN}(0,N_0)$ noise samples, where $N_0=\frac{1}{SNR}$. Also, $\mathbf{x}=\left[\mathbf{x}_1^T,\mathbf{x}_2^T,\dots,\mathbf{x}_K^T\right]^T$ is a $KN_s\times 1$ vector representing total transmitted symbols of $K$ users, satisfying $\mathbb{E}\{\mathbf{x}\mathbf{x}^H\}=\frac{1}{KN_s}\mathbf{I}_{KN_s}$. Note that $\mathbf{x}_k$ consists of the $N_s$ symbols transmitted to the $k$-th user. $\mathbf{W}_{RF}$ is the $N_r\times N_r^{\text{RF}}$ RF combining matrix and $\mathbf{W}_{BB}$ is the $N_r^{\text{RF}}\times N_s$ baseband combining matrix for $k$-th MS.
	
	We define the baseband channel as
	\begin{equation}\label{bb_channel}
	\bar{\mathbf{H}}_k={\mathbf{W}^k}^H_{RF}\mathbf{H}_k\mathbf{F}_{RF}.
	\end{equation}
	The estimated signal can be expressed as
	\begin{align}\label{rec_sig_mu}
	y_{k,s}=&{{\mathbf{W}^k}^H_{BB}\left(s,:\right)}\bar{\mathbf{H}}_k\mathbf{F}_{BB}\left(:,k_s\right)\sqrt{{P}_{k,s}}x_{k,s}\nonumber \\ \nonumber
	&+\sum_{s'=1,s'\neq i}^{N_s}{\mathbf{W}^k}^H_{BB}(s,:)\bar{\mathbf{H}}_k\mathbf{F}_{BB}(:,k_s')\sqrt{P_{k,s'}}x_{k,s'}\\
	&+\sum_{l=1,l\neq k}^{K}\sum_{s''=1}^{N_s}{\mathbf{W}^k}^H_{BB}(s,:)\bar{\mathbf{H}}_k\mathbf{F}_{BB}(:,l_{s''})\sqrt{P_{l,s''}}x_{l,s''}\nonumber\\
	&+{\mathbf{W}^k}^H_{BB}(i,:){\mathbf{W}^k}^H_{RF}\mathbf{n}_k
	\end{align}
	where $k_s=(k-1)N_s+s$, $y_{k,s}$ is the $k_s$-th element of $\mathbf{y}$ in (\ref{rec_sig0}). The first term in (\ref{rec_sig_mu}) is our desired signal,  the second term is intersymbol interference and the third term is inter-user interference. The last term is the noise.
	
	 We define $\mathbf{H}_k$  as
	\begin{align}\label{H_k}
	    \mathbf{H}_{k}=\left[\sqrt{\beta_{k1}}\mathbf{H}_{k1} \dots \sqrt{\beta_{kM}}\mathbf{H}_{kM}\right].
	\end{align}
	where $\beta_{kj}$ is a real-valued  nonnegative number which represents the large-scale fading effect between the $i$-th RAU at the receiver and $j$-th  RAU  at  the  transmitter. Note that subchannel matrix $\mathbf{H}_{kj}$ is defined as (\ref{Hij}).
	
	By modifying Algorithm 1 in \cite{XWu2018}, a two-stage hybrid beamforming is being used here to eliminate the inter-user interference. This approach maximizes the
sum-rate of the communication system  based on the two-stage approach in
massive MIMO
with double the least number of RF chains
(the least number of
RF chains is equal to the number of streams to be transmitted), i.e., $N_t^{RF}=2KN_s$ and $N_r^{RF}=2N_s$. The details of the beamforming can be found in Appendix A.

Based on (\ref{rec_sig_mu}) and by using the optimum precoders and combiners, the system input-output relation at the $m$-th time instant for the $k$-th user can be written as
\begin{align}\label{rec_sig_mu2}
    y_{k,s}^m= \sqrt{P_{k,s}}\sigma_{k,s}x_{k,s}^m + \tilde{n}_{k,s}^m
\end{align}
where $\tilde{n}_{k,s}^m={\mathbf{W}^k}^H_{BB}(s,:){\mathbf{W}^k}^H_{RF}\mathbf{n}_k^m$ and $\sigma_{k,s}$ is the $s$-th diagonal element of $\bar{\mathbf{\Sigma}}_k$, where $\bar{\mathbf{\Sigma}}_k$ is calculated by using the SVD of the matrix $\mathbf{H}_{k,\text{total}}$:
\begin{align}\label{svd_H_t}
    \mathbf{H}_{k,\text{total}}={\mathbf{W}^k}^H_{BB}{\mathbf{W}^k}^H_{RF}{\mathbf{H}}_{k}{\mathbf{F}}_{RF}{\mathbf{F}}_{BB}=\bar{\mathbf{U}}_k\bar{\mathbf{\Sigma}}_k\bar{\mathbf{V}}_k^H
\end{align}
\section{Diversity Gain and PEP Analysis for Multi-User Scenario}
In this section, we investigate using  BICMB for a multi-user scenario to increase the diversity gain while transmitting more than one data stream per user through the channel.

The inter-user interference was eliminated in Section II by using a hybrid beamforming method for multiple users. After beamforming, the pairwise error probability (PEP) can be used in a similar way to \cite{Sedighi2019} to find the upper bound for the error probability. Since this work is only concerned with high SNR regimes, we assume uniform power allocation for matrix $\mathbf{P}$. The change in achievable information rate in this case is negligible. By this assumption, without loss of generality, we assume $\mathbf{P}=\mathbf{I}_{KN_s}$.

 	\begin{theorem}\label{th2}
		When $N_t$ and $N_r$ are sufficiently large, the downlink transmission in a massive MIMO multiuser system can achieve a diversity gain for each user equal to
		\begin{align}
	        {D_{G,i}}=\frac{M^2}{\sum_{j=1}^{M}L_{ij}^{-1}}
		\end{align}
		$\text{ for } i=1,\dots,K.$
	\end{theorem}
	\begin{proof}
	The proof is similar to Theorem 1 in Section III. The argument in (\ref{ML_bit})-(\ref{sing_values}) remains the same with $y_{k,s}^m$ replacing $y_{k,s}$, $\sigma_{k,s}$ replacing $\lambda_s$, $\tilde{x}_{k,s}^m$ replacing $\tilde{x}_{k,s}$, and $\hat{x}_{k,s}^m$ replacing $\hat{x}_{k,s}^m$.
		Let us define
		\begin{align} \label{theta_MU}
		\Theta_k \triangleq& \sum_{s=1}^{L_{t}}\sigma_{k,s}^2=||\mathbf{H}_{k,\text{total}}||_F^2\nonumber \\
		=&\text{tr}\left({\mathbf{W}^k}^H_{BB}{\mathbf{W}^k}^H_{RF}{\mathbf{H}}_{k}{\mathbf{F}}_{RF}{\mathbf{F}}_{BB}{\mathbf{F}}_{BB}^H{\mathbf{F}}_{RF}^H{\mathbf{H}}_{k}^H{\mathbf{W}^k}_{RF}{\mathbf{W}^k}_{BB}\right)\nonumber\\
		 =&\text{tr}\left(\bar{\mathbf{\Sigma}}_k\bar{\mathbf{\Sigma}}_k^H\right)=\beta\frac{N_t}{N_r}\sum_{j=1}^{M}\text{tr}\left(\mathbf{\Lambda}_{kj}\mathbf{\Lambda}_{kj}^H\right)=\beta\frac{N_t}{N_r}\sum_{j=1}^{M}\sum_{s=1}^{L_t}\lambda_{kjs}^2
		\end{align}
where $\bar{\mathbf{\Sigma}}_k$ is defined in (\ref{svd_H_t}) and $\mathbf{H}_{kj}=\mathbf{A}_{kj}\mathbf{\Lambda}_{kj}\mathbf{B}_{kj}^H$ is the SVD of the matrix $\mathbf{H}_{kj}$. Note that due to the similarity between
the hybrid beamformer matrices at the transmitter and the
receiver, same design algorithms are applicable to both sides. Therefore, as mentioned in  \cite{XWu2018, Payami2016}, by choosing the optimum precoders, ${\mathbf{V}_k^{1:N_s}}^H\mathbf{F}_{\text{RF}}^{\text{opt}}\mathbf{F}_{\text{BB}}^{\text{opt}}{\mathbf{F}_{\text{BB}}^{\text{opt}}}^H{\mathbf{F}_{\text{RF}}^{\text{opt}}}^H\mathbf{V}_k^{1:N_s}=\mathbf{I}_{N_s}$ for $k=1,\dots,K$, where $\mathbf{H}_k=\mathbf{U}_k\mathbf{\Sigma}_k\mathbf{V}^H_k$ is the SVD of the channel matrix $\mathbf{H}_k$. Same procedure can be applied to the combiner part.

		Theorem 3 in \cite{Ayach2012} implies that the singular values of $\mathbf{H}_{kj}$ converge to $\sqrt{\frac{N_rN_t} {L_{kj}}}\left|\alpha^{l}_{kj}\right|$ in descending order when the number of antennas at the transmitter and the receiver goes to infinity.
		
		By using the singular values of $\mathbf{H}_{kj}$, (\ref{theta_MU}) can be rewritten as
		\begin{align}\label{Psi_def_mu}
		\Theta_k=\beta N_t^2\sum_{j=1}^{M}\underbrace{\frac{1}{L_{kj}}\sum_{l=1}^{L_{kj}}\left|\alpha_{kj}^l\right|^2}_{\Psi_{kj}}.
		\end{align}
	
		It can be seen easily that $\Psi_{kj}$ in (\ref{Psi_def_mu}) has Gamma distribution with shape $\kappa_{kj}=L_{kj}$ and scale $\theta_{kj}=2L_{kj}^{-1}$, i.e., $\Psi_{ij} \sim \mathcal{G}(L_{kj},2L_{kj}^{-1})$\cite{Hogg1978}.
		One can use the Welch-Satterthwaite equation to calculate an approximation to the degrees of freedom of $\Theta_k$ (i.e., shape and scale of the Gamma distribution) which is a linear combination of the independent random variables $\Psi_{kj}$ \cite[p.4.1-1]{Satterth1946}, \cite{Massey}
		\begin{align}\label{shape_mu}
		\kappa_k&=\frac{\left(\sum_{j}\theta_{kj}\kappa_{kj}\right)^2}{\sum_{i,j}\theta_{ij}^2\kappa_{kj}}=\frac{M^2}{\sum_{j}L_{kj}^{-1}},\\\label{scale_mu}
		\theta_k&=\frac{\sum_{j}\theta_{kj}^2\kappa_{kj}}{\sum_{j}\theta_{kj}\kappa_{kj}}=\frac{\sum_{j}L_{kj}^{-1}}{M}.
		\end{align}
		
		%%%%%%%%%%
		%%%%%%%%%%
		By following (\ref{PEP_exp2_single_user})-(\ref{BICMB_Pb_single_user}) with $\Theta_k$ replacing $\Theta$, $\theta_k$ replacing $\theta$, and $\kappa_k$ replacing $k$, we have
		\begin{align} \label{MU_bound}
		P_b \leq\frac{1}{k_c} \sum_{d=d_{\text{min}}}^{\infty}\sum_{i=1}^{W_I(d)} \frac{1}{2}\left(\theta\frac{d_{\text{free}}^2\alpha_{\text{min}}N_sN_t}{4L_t}SNR\right)^{-\kappa_k}.
		\end{align}
		
		The SNR component has a power of $-\kappa_k$ for all summations. Hence, BICMB achieves full diversity order of
		\begin{align}\label{mu_dg2}
		D_{G,k}=\kappa_k=\frac{M^2}{\sum_{j}L_{kj}^{-1}}
		\end{align}
	which is independent of the number of spatial streams transmitted.
	\end{proof}

	\begin{remarks}
		Theorem 1 implies that each MS's diversity gain is different than that of the other user and depends on the large scale fading coefficients and number of propagation paths for each user. It can be seen easily that the diversity gain is independent of the number of users.
	\end{remarks}
	\begin{remarks}
		Under the case where $N_t$ and $N_r$ are sufficiently large  and assuming that $L_{kj}=L$ and $\beta_{kj}=\beta$ for any $k$ and $j$, it can be seen easily that the distributed massive MIMO system can achieve a diversity gain
		\begin{align}
		D_{G,k}=ML
		\end{align}
		which is independent of the number of users.
	\end{remarks}

	\section{Simulation Results}
	\subsection{Single-User Scenario}
	In the simulations, the industry standard 64-state 1/2-rate (133,171) $d_{\text{free}}=10$ convolutional code is used. For BICMB, we separate the coded bits into different substreams of data and a random interleaver is used to interleave the bits in each substream. We assume that the number of RF chains in the receiver and transmitter are twice the number of data streams \cite{Sohrabi2016} (i.e., $N_t^{\text{RF}}=N_r^{\text{RF}}=2N_s$) and each scale fading coefficient $\beta_{ij}$ equals $\beta = -20$ dB (except for Fig. \ref{diff_g}). For the sake of simplicity, only ULA array configuration with $d=0.5$ is considered at RAUs. For Fig. \ref{no_intlv}--\ref{co}, Binary Phase Shift Keying (BPSK) modulation is employed for each data stream. For Fig. \ref{diff_g} information bits
are mapped onto 16 Quadrature Amplitude Modulation (QAM)
symbols in each subchannel.

\begin{figure}[!t]
	\centering
	\includegraphics[width=.70\textwidth]{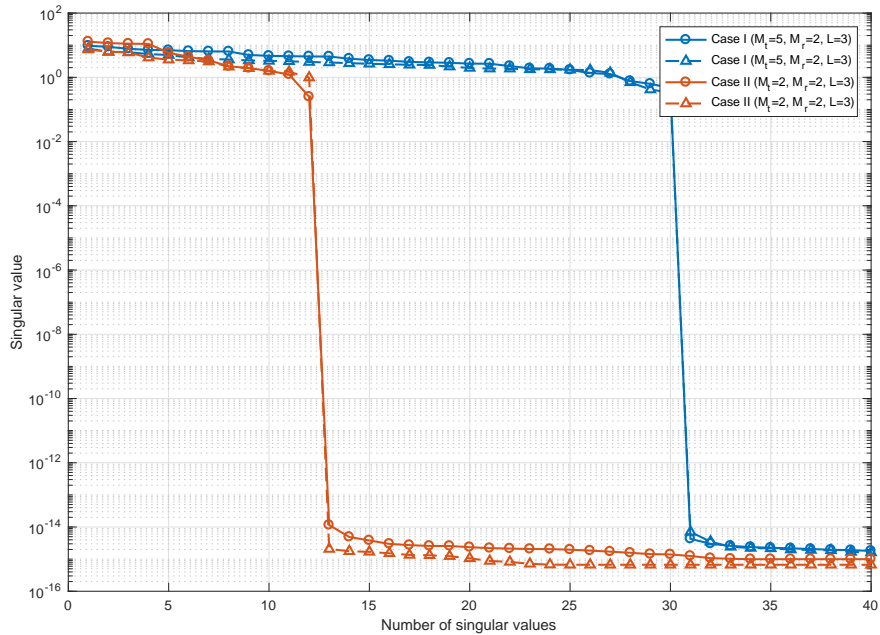}
	\caption{Singular values of the sparse mm-Wave channel with $N_t=100$ and $N_r=50$.}
	\label{singular}
\end{figure}	
	Two different cases are simulated in Fig. \ref{singular}. In Case I, the rank of the channel is $\text{rank}(\mathbf{H})=M_tM_rL=30$. For the first scenario, which is shown with circle markers, $N_t=2N_r=100$, while in the second scenario shown with triangle markers, $N_t=2N_r=400$. It can be seen from Fig. \ref{singular} that that the number of singular values of the mm-Wave channel is independent from the number of antennas at both transmitter and receiver side. Same result can be seen with Case II. Hence, there are only limited subchannels which can be used to transmit the data. The number of available subchannels $L_t=\sum_{i=1}^{M_r}\sum_{j=1}^{M_t}L_{ij}$ which is the rank of the channel $\mathbf{H}$ and is independent of the number of antennas in RAUs in both transmitter and receiver side.
	
\begin{figure}[!t]
	\centering
	\includegraphics[width=.70\textwidth]{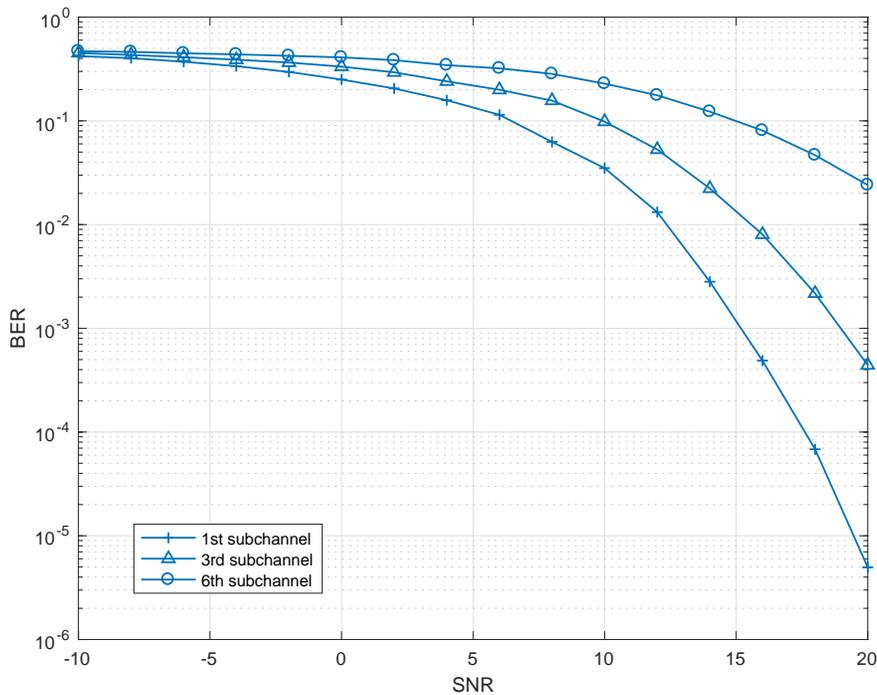}
	\caption{BER with respect to $\text{SNR}$ with $N_t=100$, $N_r=50$, $M_t=2$, $M_r=2$  and $L=2$ for $N_s=6$.}
	\label{no_intlv}
\end{figure}

Fig. \ref{no_intlv} illustrates the importance of the interleaver design. A random interleaver is used such that consecutive coded bits are transmitted over the same subchannel. Consequently, an error on the trellis occurs over paths that are spanned by the worst channel and the diversity order of coded multiple beamforming approaches to that of uncoded multiple beamforming with uniform power allocation. In other words, the BER performance decreases when the interleaving design criteria are not met.

\begin{figure}[!t]
		\centering
		\includegraphics[width=.70\textwidth]{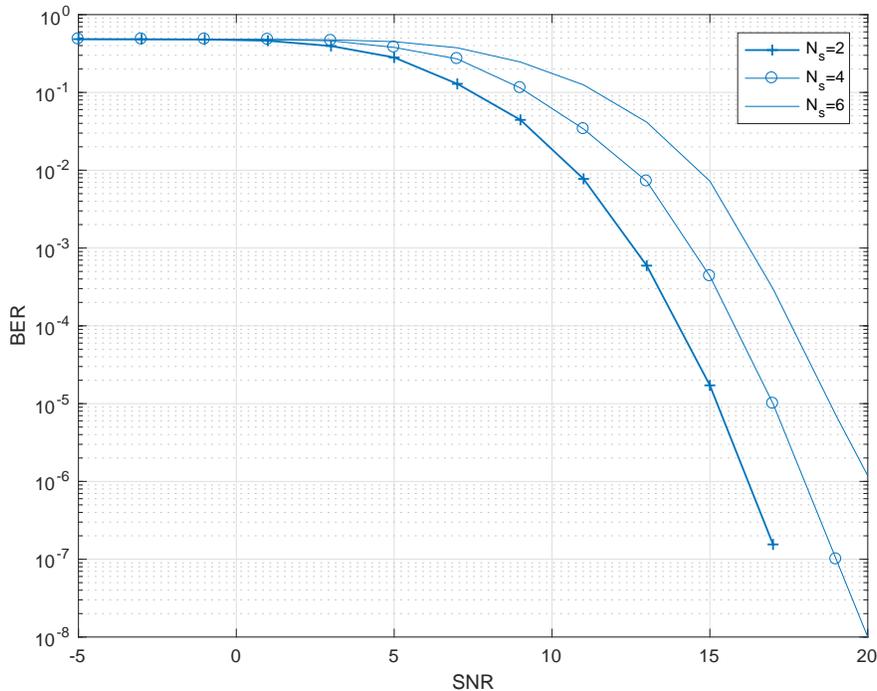}
		\caption{BER with respect to $\text{SNR}$ with $N_t=100$, $N_r=50$, $M_t=3$, $M_r=1$  and $L=2$ for different values of $N_s$.}
		\label{su_dg8}
	\end{figure}
On the other hand, as we expect from (\ref{su_dg2}), changing the number of streams $N_s$ should not change the diversity gain, i.e., the slope of the BER curve in high SNR. As it can be seen from Fig. \ref{su_dg8}, the slope does not change by changing the number of data streams. Hence, one can get the same diversity gain by using the maximum number of data streams available ($L_t$).
	
\begin{figure}[!t]
	\centering
	\includegraphics[width=.70\textwidth]{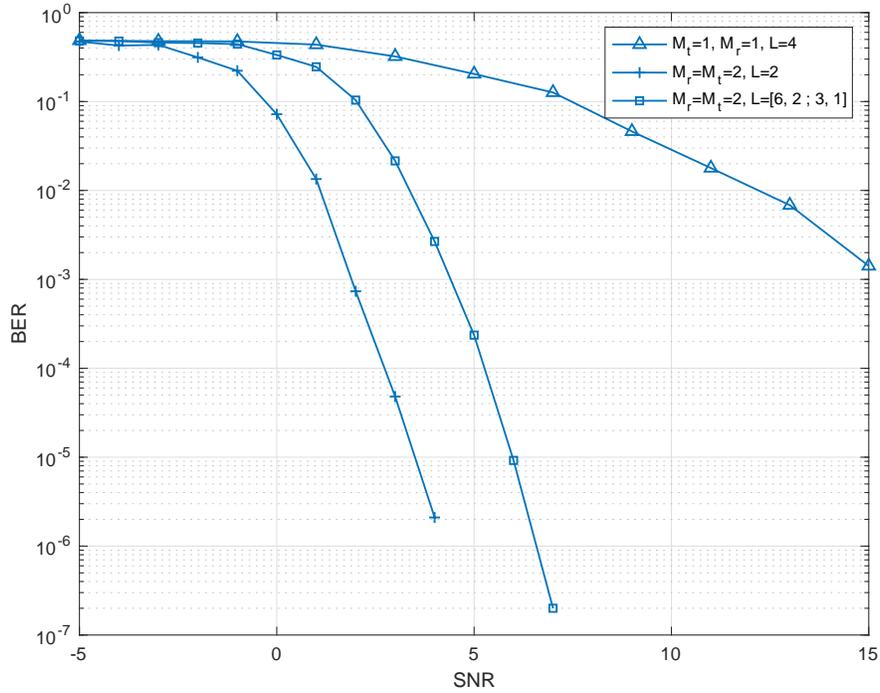}
	\caption{BER with respect to $\text{SNR}$ with $N_t=100$, $N_r=50$ and $N_s=3$.}
	\label{co}
\end{figure}
Fig. \ref{co} illustrates the results for BICMB for both co-located and distributed mm-Wave massive MIMO systems. The diversity gain for the distributed system outperforms the co-located system, even though the channel in the co-located system has richer scattering (the number of propagation paths in the co-located system is twice as the distributed system). Also, as it can be seen from the figure, the curves for the distributed systems are parallel to each other, especially for the high-SNR region, which can be confirmed by (\ref{su_dg2}). Note, for distributed systems, when $\beta_{ij}=\beta$, $2 \times 2 \times 2 = (2\times 2)^2 / (6^{-1}+2^{-1}+3^{-1}+1^{-1})$ as in (\ref{su_dg2}).

\begin{figure}[!t]
	\centering
	\includegraphics[width=.70\textwidth]{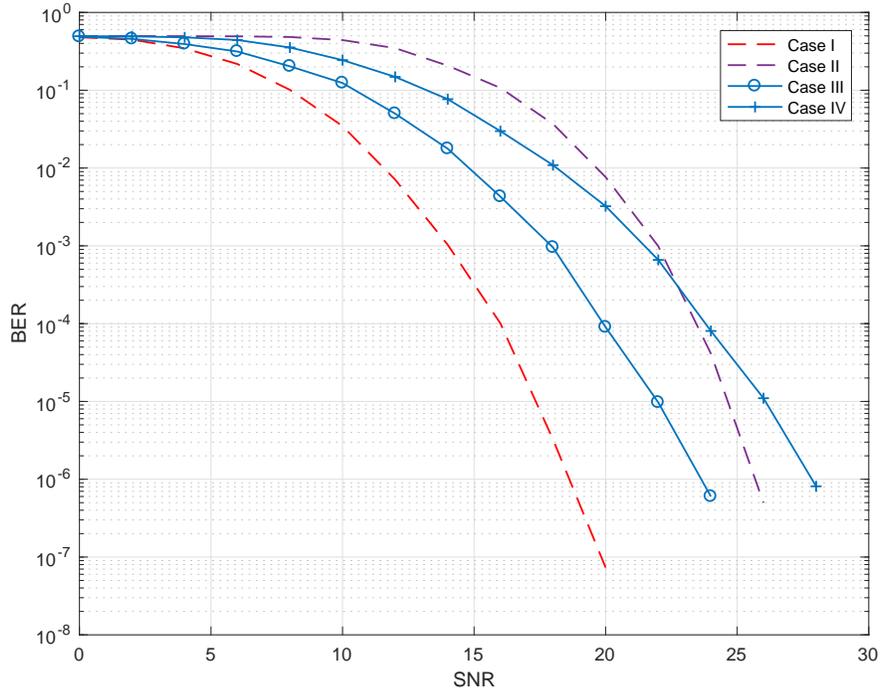}
	\caption{BER with respect to $\text{SNR}$ with $N_t=100$, $N_r=50$ and $N_s=1$.}
	\label{diff_g}
\end{figure}
Fig. \ref{diff_g} illustrates the effect of the large scale fading coefficient on the diversity gain. Despite other simulations, we consider inhomogeneous large scale fading coefficients. When $M_r=M_t=2$, $L_{ij}=L=2$, $N_t=2N_r=100$ and $N_s=1$ three different cases are simulated. Let $\mathcal{\mathbf{B}}=\left[\beta_{ij}\right]$ where $\beta_{ij}$ expressed in dB, as the large scale fading coefficient matrix. We used the following $\mathbf{B}$ in the simulations:
\begin{align*}
    &\mathbf{B}_1=\begin{bmatrix}
    -20   & -20 \\
    -20  & -20
\end{bmatrix},
    \mathbf{B}_2=\begin{bmatrix}
    -25   & -25 \\
    -25  & -25
\end{bmatrix},\\
    &\mathbf{B}_3=\begin{bmatrix}
    -20   & -35 \\
    -35  & -20
\end{bmatrix},
    \mathbf{B}_4= -20 .
\end{align*}

As it can be seen from Fig. \ref{diff_g}, when the system is homogeneous, the diversity gain remains the same. Case I and Case II, have the same slope in high SNR, which is expected. In Case III, when the system is inhomogeneous, the diversity gain decreases. By using (\ref{su_dg2}), one can easily see that Case III has approximately the same diversity gain as a system with $M_r=M_t=1$ and $L=4$, i.e., $D_G=4$, which is depicted in Case IV.

	\subsection{Multi-User Scenario}
	The assumptions remains the same in these simulations unless otherwise stated. We assume that each scale fading coefficient $\beta_{kj}$ equals to $\beta = - 20$ dB (except for Fig. \ref{attenuation}). In the simulations, information bits
are mapped onto 16 quadrature amplitude modulation (QAM)
symbols in each subchannel.

\iffalse
      Fig. \ref{singulars} depicts a simple case of (\ref{theta}), where there are two different users, i.e., $K=2$ and each user has four different subchannels, which gives us a total number of eight subchannels. On the top left, the singular values of the processed channel and on the top right, the singular values of the actual channel can be seen. On the bottom of the figure, the proportion of these two figures for each subchannel is constant, which means that the effect of using beamforming which results in the processed channel can be recovered at the user side.
		\begin{figure}[h!]
		\centering
		\includegraphics[width=0.70\textwidth]{figures/singulars.eps}
		\caption{Relation between the singular values of the actual channel $\mathbf{H}_k$ and the processed channel $\mathbf{H_{k,\text{total}}}$ with $N_t=512$, $N_r=64$, $M=3$ and $K=2$}
		\label{singulars}
		
	\end{figure}
	\fi

\begin{figure}[!t]
		\centering
		\includegraphics[width=0.70\textwidth]{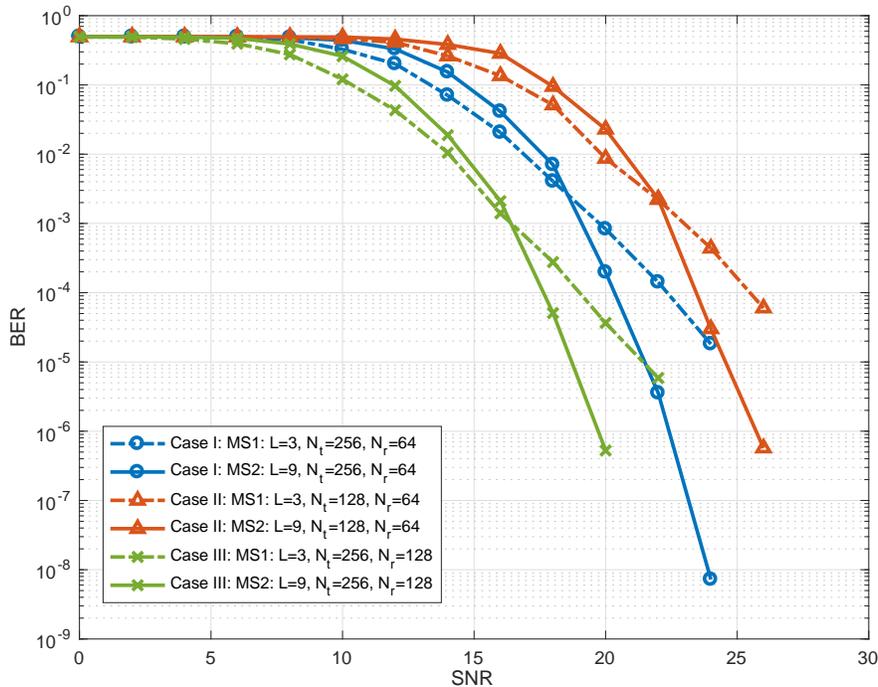}
		\caption{BER with respect to SNR for two cases with different number of antennas in each RAU.   $M=2$, $K=2$, $N_s=3$}
		\label{ber_antenna}
		
\end{figure}
Fig. \ref{ber_antenna} illustrates BER with respect to SNR for three different cases. In each case, two different MSs are being served by the BS. The BS transmits three data streams to each MS. For the first MS, the number of propagation paths $L$ in each case is $L=3$, while for the second MS $L=9$ for all subchannels. By comparing Case I with $N_t=256$ and Case II $N_t=128$ where circle markers represent Case I and triangle markers are for Case II, one can easily see that doubling the number of antennas at the BS has no effect on the slope of the BER, i.e., the diversity gain in high SNR. This confirms (\ref{mu_dg2}) where the diversity gain is independent of the number of antennas at the BS. The independence of the (\ref{mu_dg2}) from the number of antennas at the MS side can be seen by comparing Case III with Case I, where both of them have the same number of antennas at the transmitter, but in Case III, the number of antenna elements at the MS side is twice as Case I. Note that the markers of BER of the users in Case III are a cross sign (x).
	
\begin{figure}[!t]
		\centering
		\includegraphics[width=0.70\textwidth]{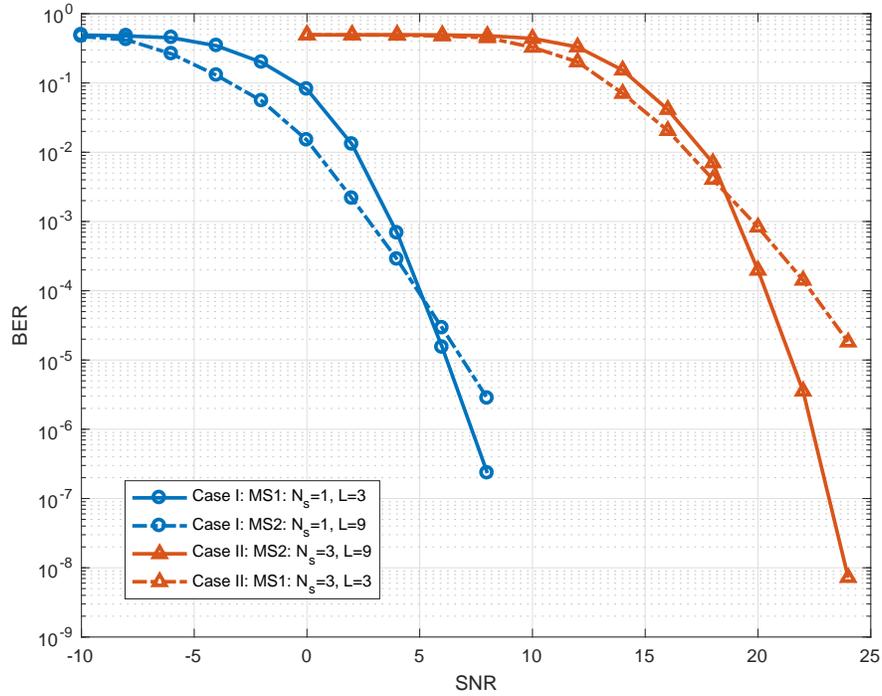}
		\caption{BER with respect to SNR for two cases with different number of data streams sent through the channel.   $M=2$, $K=2$, $N_t=256$ and $N_r=64$}
		\label{ber_stream}
\end{figure}
Similar to Fig. \ref{ber_antenna}, the BS serves two different MSs in two cases in Fig. \ref{ber_stream}. In Case I, each MS only receives one data stream, while in Case II, each MS receives three data streams. When there is no BICMB, one can get the maximum diversity gain by only sending one data stream through the channel. This can be used as a benchmark to compare the diversity gain when the number of data streams increases. It can be seen that with BICMB by sending more data streams through the channel, the slope of the BER curve does not change in high SNR. Hence, one can get the same diversity gain by transmitting maximum number of data streams, i.e., rank of the channel through the channel.
	
\begin{figure}[!t]
		\centering
		\includegraphics[width=0.70\textwidth]{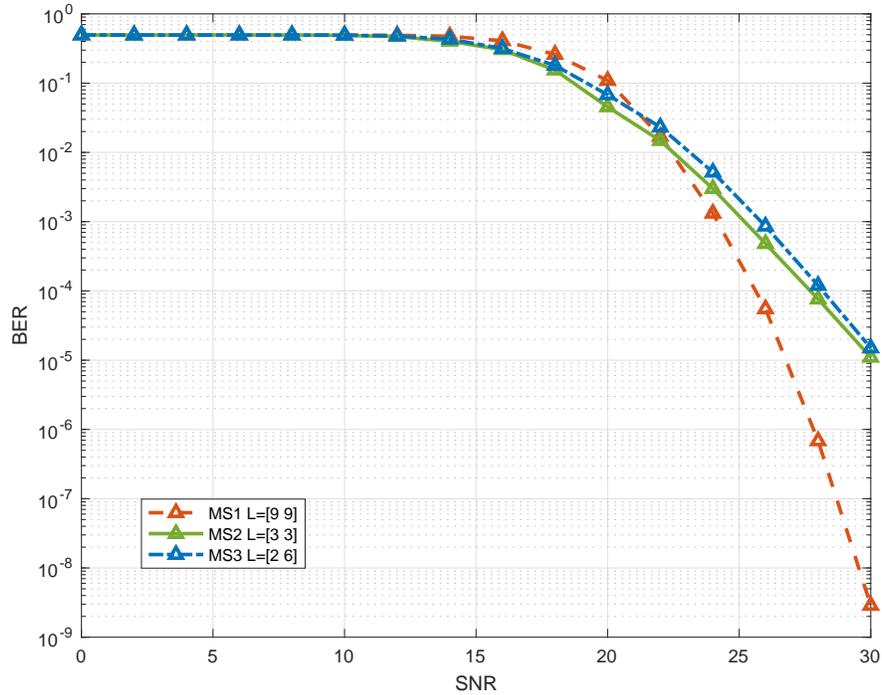}
		\caption{BER with respect to SNR when the number of user increases from two to three when each user receives three data streams.  $L=[l_1 l_2]$ means that the number of propagation paths to the user from the first RAU is $l_1$ and same for the $l_2$ and the second RAU. $M=2$, $K=3$, $N_t=256$ and $N_r=64$}
		\label{new_user}
\end{figure}
Comparing Fig. \ref{ber_antenna} or Fig. \ref{ber_stream} with  Fig. \ref{new_user} shows that by increasing the number of MSs in the system, the diversity gain does not change. Also, one can check (\ref{mu_dg2}) for the second and the third user to see that they have both the same diversity gain.
	
\begin{figure}[!t]
		\centering
		\includegraphics[width=0.70\textwidth]{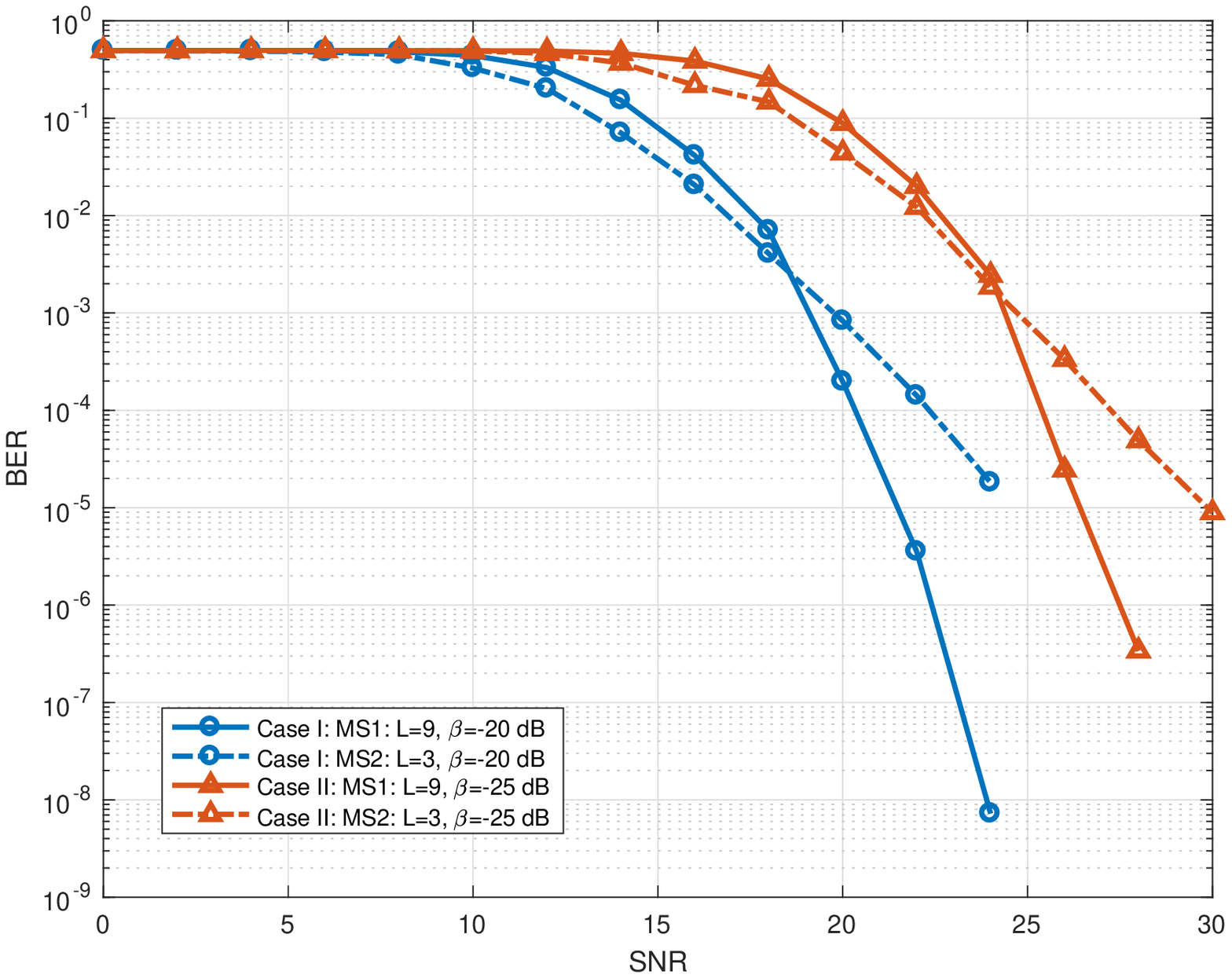}
		\caption{BER with respect to SNR with different large-scale fading coefficient $\beta$. $M=2$, $K=2$, $N_t=256$ and $N_r=64$}
		\label{attenuation}
		
\end{figure}
		It is expected from (\ref{mu_dg2}) that the diversity gain is independent of large-scale fading coefficient $\beta_{kj}$ when the large-scale fading coefficient $\beta_{kj}$ is constant, i.e., $\beta_{kj}=\beta$. Fig. \ref{attenuation} illustrates two different cases. In the first case, $\beta_{kj}=\beta=-25$ dB and in the second case, the value of $\beta$ increases to $\beta=-20$ dB. In both cases, we are transmitting three different data streams for each user. Also, two users are being served in each case. As we expect, the diversity gain remains the same when the large-scale fading coefficient remains constant for all subchannels.
	
\section{Conclusion}
	In this paper we analyzed BICMB in mm-Wave massive MIMO systems for both single-user and multi-user scenarios. BICMB achieves full spatial diversity of $\frac{\left(\sum_{i,j}\beta_{ij}\right)^2}{\sum_{i,j}\beta_{ij}^2L_{ij}^{-1}}$ over $M_t$ RAU transmitters and $M_r$ RAU receivers in the single-user scenario. This means, by increasing the number of RAUs in the distributed system with BICMB, one can increase the diversity gain and multiplexing gain. As it can be seen from the diversity gain formula for the single-user system, the value of diversity gain is independent of the number of antennas in each RAU for both transmitter and receiver. A special case of the diversity gain where $L_{ij}=L$ and $\beta_{ij}=\beta$ would be $M_rM_tL$ which is similar to the diversity gain of a convential MIMO system.  In a multi-user system, BICMB achieves full spatial diversity of $\frac{M^2}{\sum_{j}L_{kj}^{-1}}$ over $M$ RAU transmitter for the $k$-th user. This means one can increase the diversity gain for all users in such a system by increasing the number of RAUs at the transmitter side. Another result is that the diversity gain is independent of the number antennas in both transmitter and receiver side. In a special case when $L_{kj}=L$, the diversity gain is $ML$ which looks like the single-user scenario in \cite{Sedighi2019} when ${M_r=1}$ and $M_t=M$.

	\begin{appendices} \section{Hybrid Beamforming for Multi-User
Massive MIMO Systems}
	In this appendix, the hybrid block diagonalization beamforming for multi-user scenario is summarized based on \cite{XWu2018}. First, by using (\ref{H_k}) and SVD one can define
	\begin{align}\label{app_eq1}
	\frac{1}{\sqrt{N_t}}\mathbf{H}_k=\mathbf{U}_k\mathbf{\Sigma}_k\mathbf{V}_k^H
	\end{align}
	and
		\begin{align}\label{app_eq2}
	\frac{1}{\sqrt{N_t}}\mathbf{H}_{\text{comp}}=\mathbf{U}_{\text{comp}}\mathbf{\Sigma}_{\text{comp}}\mathbf{V}_{\text{comp}}^H
	\end{align}
where
\begin{align}\label{app_eq3}
    \mathbf{H}_{\text{comp}}=\mathbf{W}_{\text{RF}}^H\mathbf{H}=\begin{bmatrix}
    {\mathbf{W}_{RF}^1}^H & 0 & \dots &  0 \\
    0 & {\mathbf{W}_{RF}^2}^H & 0 & \dots &  0 \\
   \vdots & \vdots & \ddots & \vdots \\
   0 & 0 & \dots &   {\mathbf{W}_{RF}^K}^H
\end{bmatrix} =\begin{bmatrix}
    \mathbf{H}_{1}  \\
   \vdots  \\
   \mathbf{H}_{K}
\end{bmatrix}
\end{align}
By using these definitions here and the material in Section II, a closed-form solution for hybrid beamforming can be derived as Algorithm 1. Here, the number of RF chains is double the least number of RF chains, i.e., $N_{\text{r}}^{\text{RF}}=2N_s$ and $N_{t}^{\text{RF}}=2KN_s$. After calculating the beamforming matrices by Algorithm 1, (23) in \cite{XWu2018} is used to  transform the scheme to the constrained case mentioned earlier.
\begin{algorithm}[H]
\DontPrintSemicolon

  \KwIn{$\mathbf{H}$}
  \For{$k=1$ to $K$}
        {
        	Calculate $\mathbf{W}_{\text{RF}}^k=\mathbf{U}_k(:,1:M_\text{MS})$ where $\mathbf{U}_k$ is defined as (\ref{app_eq1})
        }
  compute $\mathbf{W}_{\text{RF}}$ based on (\ref{app_eq3}) \\
      By using (\ref{app_eq2}), compute $\mathbf{F}_{RF}=\mathbf{V}_{\text{comp}}(:,1:N_t^{\text{RF}})$\\
    \For{$k=1$ to $K$}
        {
  Calculate the baseband channel for the $k$-th user as (\ref{bb_channel})
  }
  Compute $\mathbf{F}_{\text{BB}}$ and $\mathbf{W}_{\text{BB}}$ by using the scheme described in Section IV in \cite{XWu2018}, then normalize each column of $\mathbf{F}_{\text{BB}}$ as $\mathbf{F}_{\text{BB}}(:,i)=\frac{\text{F}_{\text{BB}}(:,i)}{||\text{F}_{\text{RF}}\text{F}_{BB}(:,i)||_F}$

\KwOut{$\mathbf{F}_{\text{RF}}, \mathbf{F}_{\text{BB}}, {\mathbf{W}_{\text{RF}}}^k,({\mathbf{W}_{\text{BB}}^k})_{k=1:K}$}
\caption{Hybrid Block Diagonalization Beamforming for multi-user Scenario}
\end{algorithm}

 \end{appendices}
\newpage
	\bibliographystyle{IEEEtran}
	\bibliography{Globecomm2019}
\end{document}